\documentclass[a4paper]{easychair}
\usepackage{amssymb}
\usepackage{amsmath}
\usepackage{amsthm}
\usepackage{tikz}
\usetikzlibrary{arrows,decorations.pathmorphing,backgrounds,fit,positioning}
\usepackage{xspace}

\newtheorem{theorem}{Theorem}
\newtheorem{observation}[theorem]{Observation}
\newtheorem{example}[theorem]{Example}

\newcommand\isafor{\textsf{Isa\kern-0.2exF\kern-0.2exo\kern-0.2exR}\xspace}
\newcommand\ceta{\textsf{C\kern-0.2exe\kern-0.5exT\kern-0.5exA}\xspace}
\newcommand\CSI{\textsf{CSI}\xspace}

\newcommand{\doi}[1]{}

\newcommand\sem[1]{[\![#1]\!]_{\AAA,\alpha}}
\newcommand\semz[1]{[\![#1]\!]_{\AAA,\alpha_d}}
\newcommand\VV{{\cal V}}
\newcommand\TT{{\cal T}}
\newcommand\LL{{\cal L}}
\newcommand\AAA{{\cal A}}
\newcommand\QQ{{\cal Q}}
\newcommand\final{{\cal Q}_{f}}
\newcommand\FF{{\cal F}}
\newcommand\RR{{\cal R}}
\newcommand\Ur{{\cal U}_r}
\newcommand\Uz{{\cal U}_0}
\newcommand\SSS{{\cal S}}
\newcommand\CR{\mathit{CR}}
\newcommand\CP{\mathit{CP}}
\newcommand\nCR{\neg \CR}

\newcommand\njin[3]{\mathit{NJ}_{#1}(#2,#3)}
\newcommand\njo[3]{\njin{#1,#1}{#2}{#3}}
\newcommand\nj[4]{\njin{#1,#2}{#3}{#4}}
\newcommand\m[1]{\mathsf{#1}}
\newcommand\tcap[1]{\mathit{tcap}_{#1}}

\newcommand{\from}{\leftarrow}
\newcommand{\fromBT}[2]{\mathrel{{\vphantom{\to}^{#1}_{#2}}{\from}}}

\newcommand\rTH[1]{Theorem~\ref{#1}}
\newcommand\rObs[1]{Observation~\ref{#1}}
\newcommand\rEX[1]{Example~\ref{#1}}
\newcommand\rSC[1]{Section~\ref{sect:#1}}

\newcommand{\VCLHS}[1]{\ensuremath{\forall \ell \to r \in #1.\,\ell \notin \VV}}
\newcommand{\VCSUBSET}[1]{\ensuremath{\forall \ell \to r \in #1.\,\VV(\ell) \supseteq \VV(r)}}
\newcommand{\vclhs}[1]{\mathit{VC}_{\mathit{lhs}}(#1)}
\newcommand{\vcsubset}[1]{\mathit{VC}_{\supseteq}(#1)}

\begin{document}

%
\title{Certification of Confluence Proofs using \ceta\thanks{Supported by the
Austrian Science Fund (FWF), projects P22467 and P22767.}}
\titlerunning{Certification of Confluence Proofs using \ceta}

\author{Julian Nagele \and Ren\'{e} Thiemann}
\authorrunning{Julian Nagele and Ren\'{e} Thiemann}

\institute{
University of Innsbruck, Austria, 
\email{\{julian.nagele|rene.thiemann\}@uibk.ac.at}
}

\maketitle

\section{Introduction}
\label{sect:introduction}
\ceta was originally developed as a tool for certifying termination proofs \cite{TS09b}, which have to 
be provided as certificates in the CPF-format. Given a certificate \ceta will either answer \texttt{CERTIFIED}, 
or return a detailed error message why the proof was \texttt{REJECTED}.
Its correctness is formally proven as part of \isafor, the \textbf{Isa}belle \textbf{F}ormalization \textbf{o}f 
\textbf{R}ewriting: \isafor contains executable ``check''-functions for each formalized proof technique together with
formal proofs that whenever such a check is accepted, the technique is applied correctly. 
Isabelle's code-generator is then used to obtain \ceta.\footnote{At \url{http://cl-informatik.uibk.ac.at/software/ceta/} 
one can access \ceta, \isafor, and the CPF-format.}
By now, \ceta can also be used for certifying confluence and non-confluence proofs. In this system description, we give
an overview on what kind of proofs are supported, and what information has to be given in
the certificates. As we will see, only little information is required and so
we hope that \CSI \cite{ZFM11b} will not stay the only confluence tool that can
produce certificates. 

\section{Terminating Term Rewrite Systems (TRSs)}
\label{sect:terminating}
It is well known that confluence of terminating TRSs is decidable by checking
joinability of all critical pairs. The latter can be decided by reducing both terms of
a critical pair to arbitrary normal forms and then checking if these are equal.
This technique is also supported in \ceta, where in the certificate one just has to 
provide the termination
proof and \ceta automatically constructs all critical pairs and checks their joinability by
rewriting to normal forms. 
Alternatively one can also specify to check joinability by an automatic breadth-first search. 
Finally one can completely provide the joining sequences for all critical pairs in the certificate. 
Although the latter results in more verbose certificates, which are harder to produce,
they are faster to check as no search is required for certification.
For example, for $\RR = \RR_{\m{ack}} \cup $ $\{\m f(x) \to
  x, \m a \to \m{ack}(1000,1000), \m a \to \m f(\m{ack}(1000,1000))\}$, where $\RR_{\m{ack}}$
  is a convergent TRS for the Ackermann-function, all critical pairs are joinable, 
  but rewriting to normal form won't work.
  
  
\section{Certificates for Confluence}
\isafor contains formalizations of two techniques
that ensure confluence and do not demand termination: strongly closed and linear TRSs as well as weakly orthogonal TRSs 
are confluent.

For the latter, the certificate only consists of the statement that the TRS is weakly
orthogonal, which is a syntactic criterion that can easily be checked by \ceta. 
For the former criterion, the interesting part is to ensure that a given TRS $\RR$ is strongly closed,
i.e., for every critical pair $(s,t)$ there are terms $u$ and $v$ such that 
$s \to^*_\RR u \fromBT{=}{\RR} t$ and $s \to^=_\RR v \fromBT{*}{\RR} t$. Clearly,
rewriting to normal forms is of little use here, so we just offer a breadth-first search
in \ceta. In the certificate one just has to provide a bound on the length of the joining
derivations. The reason for requiring the explicit bound is that
in Isabelle all functions have to be total.\pagebreak In contrast to \rSC{terminating},
here $\RR$ is not necessarily terminating,
and thus an unbounded breadth-first search might be non-terminating, whereas an explicit bound on the
depth easily ensures totality.

At this point, let us recall our notions of TRSs and critical pairs:
as usual a TRS $\RR$ is just a set of rewrite rules. However we do not assume
the following standard variable conditions:
\begin{xalignat*}{2}
\vclhs\RR & = \VCLHS\RR &
\vcsubset\RR & = \VCSUBSET\RR
\end{xalignat*}
The critical pairs of a TRS $\RR$ are defined as
\[
\CP(\RR) = \{ (r\sigma,C[r']\sigma) \mid \ell \to r \in \RR,\, \ell' \to r' \in \RR,\, \ell = C[u],\, u \notin \VV,\, \mathit{mgu}(u,\ell') = \sigma\}
\]
where it is assumed that the variables in $\ell \to r$ and $\ell' \to r'$ have been renamed
apart. We do not exclude root overlaps of a rule with itself, which gives
rise to trivial critical pairs of the form $(r\sigma,r\sigma)$. Therefore, most
techniques in \isafor that rely on critical pairs immediately try to remove all trivial 
critical pairs, i.e., they consider $\{(s,t) \in \CP(\RR) \mid s \neq t\}$
instead of $\CP(\RR)$. So, in practice these additional critical pairs do not play a role. 
However, for TRSs
that do not satisfy the variable conditions they are essential. For example,
for the TRS $\RR_1 = \{\m a \to y\}$ over signature $\{\m a,\m b,\m c\}$ we have $\CP(\RR) = \{(x,y)\}$, whereas without root-overlaps
with the same rule there would be no critical pair and we might wrongly conclude confluence
via orthogonality.

The confluence criterion of weak orthogonality not only implicitly demands
$\vcsubset\RR$, but explicitly demands $\vclhs\RR$.
In contrast, none of the variable conditions is required for strongly closed and linear TRSs.
Hence, the following two TRSs are confluent via this criterion: $\RR_2 = \{x \to \m f(x), y \to \m g(y)\}$ is strongly closed as there are no critical pairs, and $\RR_3 = \{\m a \to \m f(x), \m f(x) \to \m b\}$ is strongly closed as the
only non-trivial critical pair is $(\m f(x), \m f(y))$, which is obviously joinable in one step
to $\m b$. Also $\RR_4 = \{\m a \to \m f(x), \m f(x) \to \m b, x \to \m f(\m g(x))\}$---which 
satisfies neither of the variable conditions---is strongly closed
and linear, and thus confluent. 
Similarly as for weak orthogonality, the addition of root overlaps w.r.t.\ the
same rule is essential, as otherwise the non-confluent and linear TRS $\RR_1$ would be
strongly closed.

\section{Disproving Confluence via Non-Joinable Forks}
\label{sect:disprove}

One way to disprove confluence of an arbitrary, possibly non-terminating TRS $\RR$ is to
provide a non-joinable fork, i.e., $s \to_\RR^* t_1$ and $s \to_\RR^* t_2$ such that $t_1$ and
$t_2$ have no common reduct. To certify these proofs, in \ceta we demand the concrete derivations
from $s$ to $t_1$ and $t_2$ and additionally a certificate that $t_1$ and $t_2$ are not joinable, which is clearly
the more interesting part. To this end, we generalize the notion of non-joinability to two TRSs,
which allows us to conveniently and modularly formalize several existing
techniques for non-joinability. Initially, $\RR_1 = \RR_2 = \RR$ and any change on one of the TRSs
is currently internally computed by \ceta.
\[
\nj{\RR_1}{\RR_2}{t_1}{t_2} = (\neg \exists u.\, t_1 \to_{\RR_1}^* u \wedge t_2 \to_{\RR_2}^* u)
\]
%


\subsection{Grounding}
Clearly, $\nj{\RR_1}{\RR_2}{t_1\sigma}{t_2\sigma}$ implies $\nj{\RR_1}{\RR_2}{t_1}{t_2}$ for
some arbitrary substitution $\sigma$. This substitution has to be provided in the certificate
and can be used replace each variable in $t_1$ and $t_2$ by some fresh constant.
Grounding can be beneficial for other non-joinability techniques.

\subsection{Tcap and Unification}
The function $\tcap\RR$ can approximate an upper part of a term where no rewriting with $\RR$
is possible, and thus, remains unchanged by rewriting. Hence, \pagebreak
it suffices to check that $\tcap{\RR_1}(t_1)$ is not unifiable with $\tcap{\RR_2}(t_2)$
to ensure $\nj{\RR_1}{\RR_2}{t_1}{t_2}$.
  
Since $\tcap{\RR_i}$ replaces variables by fresh ones, it is beneficial to
apply grounding beforehand \cite{ZFM11b}. To this end, \ceta computes a suitable 
grounding substitution, if some $t_i$ is not a ground term. Because of  grounding, 
this criterion fully subsumes the criterion, that two different normal forms are not joinable.
Nevertheless one can also refer to the latter criterion in certificates.

\subsection{Usable Rules for Reachability}
In \cite{Aoto} the usable rules for reachability $\Ur$ have been defined (via some 
inductive definition of auxiliary usable rules $\Uz$). They have
the crucial property that $t \to_\RR^* s$ implies $t \to_{\Ur(\RR,t)}^* s$.
This property immediately shows the following theorem.
\begin{theorem}
\label{usable}
$\nj{\Ur(\RR_1,t_1)}{\Ur(\RR_2,t_2)}{t_1}{t_2}$ implies
$\nj{\RR_1}{\RR_2}{t_1}{t_2}$.
\end{theorem}

Whereas the crucial property was easily formalized within \isafor following the 
original proof, it was actually more complicated to provide an implementation of usable rules
that turns the inductive definition of $\Uz$ into executable code. 
%
Note that we did not have this problem in previous work on
usable rules \cite{CSRT-RTA10} where we explicitly demand that the set of usable rules is provided in the certificate. However, due to our implementation of usable rules, we no longer require the
set of usable rules in the certificate.

\subsection{Discrimination Pairs}
In \cite{Aoto} term orders are utilized to prove non-joinability.
To be precise, $(\succsim,\succ)$ is a discrimination pair iff $\succsim$ is a rewrite order,
$\succ$ is irreflexive, and ${\succsim} \circ {\succ} \subseteq {\succ}$.\footnote{Note, that
unlike what is said in \cite{Aoto}, one 
does not require ${\succ} \circ {\succsim} \subseteq {\succ}$.}
We formalized the following theorem, which in combination with \rTH{usable}
completely simulates \cite[Theorem 12]{Aoto}.
\begin{theorem}

\label{discrimination pairs}
If $(\succsim,\succ)$ is a discrimination pair, ${\RR_1^{-1} \cup \RR_2} \subseteq {\succsim}$, 
and $t_1 \succ t_2$ then $\nj{\RR_1}{\RR_2}{t_1}{t_2}$.
\end{theorem}

\begin{proof}
We perform a proof by contradiction, so 
assume $t_1 \to_{\RR_1}^* u$ and $t_2 \to_{\RR_2}^* u$ and hence $t_2 \to_{\RR_1^{-1} \cup \RR_2}^* t_1$.
Then by the preconditions
we obtain $t_2 \succsim^* t_1 \succ t_2$.
Iteratively applying ${\succsim} \circ {\succ} \subseteq {\succ}$ yields $t_2 \succ t_2$ 
in contradiction to irreflexivity of $\succ$.
\end{proof}

We have also proven within \isafor that every reduction pair is a discrimination pair, and
thus one can use all reduction pairs that are available in \ceta in the certificate.

\subsection{Argument Filters}
In \cite{Aoto} it is shown that argument filters $\pi$ are useful for non-confluence
proofs. The essence is

\begin{observation}
\label{argument filter}
$\nj{\pi(\RR_1)}{\pi(\RR_2)}{\pi(t_1)}{\pi(t_2)}$ 
implies $\nj{\RR_1}{\RR_2}{t_1}{t_2}$.
\end{observation}

Consequently, one may show non-joinability by applying an argument filter
and then continue on the filtered problem.
At this point we can completely simulate \cite[Theorem 14]{Aoto}:
apply usable rules, apply argument filter, apply usable rules, apply discrimination pair.

\subsection{Interpretations}
Let $\FF$ be some signature. Let $\AAA$ be a weakly monotone $\FF$-algebra $(A,(f^\AAA)_{f \in \FF},\geq)$, i.e.,
$f^\AAA : A^n \to A$ for each $n$-ary symbol $f \in \FF$, $\geq$ is a partial order, \pagebreak
and for all $a,b,f$, $a \geq b$ implies $f^\AAA(\dots,a,\dots) \geq f^\AAA(\dots,b,\dots)$.
$\AAA$ is a quasi-model for $\RR$ iff $\sem{\ell} \geq \sem{r}$
for all $\ell \to r \in \RR$ and every valuation $\alpha : \VV \to A$.
Let $\alpha_d$ be some default valuation.

\begin{theorem}
\label{quasi}
If $\AAA$ is a quasi-model of $\RR_1^{-1} \cup \RR_2$ and $\semz{t_2} \not\geq \semz{t_1}$ then
$\nj{\RR_1}{\RR_2}{t_1}{t_2}$.
\end{theorem}

\begin{proof}
Similar as for \rTH{discrimination pairs}. Given $t_2 \to_{\RR_1^{-1} \cup \RR_2}^* t_1$
and the quasi-model condition we conclude $\semz{t_2} \geq \semz{t_1}$. 
This is an immediate contradiction to 
$\semz{t_2} \not\geq \semz{t_1}$.
\end{proof}

This proof was easy to formalize as it could reuse the formalization of semantic labeling~\cite{CSRT-RTA11},
which also includes algorithms to check the quasi-model conditions as well as 
a format for models in the certificate.
Here, \ceta is currently restricted to algebras over finite domains.
Moreover, the valuation $\alpha_d$ cannot be specified in the certificate. However,
by previously applying grounding, the choice of $\alpha_d$ does not matter any longer.

Note that in contrast to \cite[Theorem 10]{Aoto}, we only require $\semz{t_2} \not\geq \semz{t_1}$
instead of $\semz{t_2} \not\geq \semz{t_1} \wedge \semz{t_1} \geq \semz{t_2}$.
This has an immediate advantage, namely that we can derive \cite[Corollary 6]{Aoto} as a consequence:
instantiate $\geq$ by equality, then weak monotonicity is always guaranteed, the quasi-model condition
becomes a model condition, and $\semz{t_2} \not\geq \semz{t_1}$ is equivalent to
$\semz{t_1} \neq \semz{t_2}$. Moreover, the usable rules can easily be integrated as a preprocessing
step in the same way as we did for discrimination pairs.

Further note that \cite[Corollary 6]{Aoto} can also simulate \cite[Theorem 5]{Aoto}, by just
taking the quotient algebra. Therefore, by Theorems~\ref{usable}, \ref{discrimination pairs}, and 
\ref{quasi}, and \rObs{argument filter}
we can now simulate all non-joinability criteria of \cite{Aoto} and \ceta can also certify all 
example proofs of \cite{Aoto}.

\subsection{Tree Automata}

A bottom-up tree automaton $\AAA$ is a quadruple $(\QQ,\FF,\Delta,\final)$ 
with states $\QQ$, signature $\FF$,
transitions $\Delta$, and final states $\final$, and $\LL(\AAA) \subseteq \TT(\FF)$ 
denotes the accepted regular tree language. We say that $\AAA$ is closed under $\RR$ if
$\{t \mid s \in \LL(\AAA), s \to_\RR t\} \subseteq \LL(\AAA)$.

\begin{observation}
\label{automata}
Let $\AAA_1$ and $\AAA_2$ be tree automata.
If $t_i \in \LL(\AAA_i)$ and $\AAA_i$ is closed under $\RR_i$ for $i = 1,2$, and
$\LL(\AAA_1) \cap \LL(\AAA_2) = \varnothing$ then $\nj{\RR_1}{\RR_2}{t_1}{t_2}$.
\end{observation}

For checking these non-joinability certificates, \ceta implemented standard tree automata algorithms
for membership, intersection, and emptiness. The most difficult part is checking whether
$\AAA$ is closed under $\RR$ for some $\AAA$ and $\RR$.
Here, \ceta provides three alternatives. One can refer to Genet's criterion 
of compatibility, or use the more liberal condition of state-compatibility \cite{BFRT-LATA14}, which 
requires an additional compatibility relation in the certificate, or one can just refer to
the decision procedure \cite{BFRT-LATA14}, which currently requires a deterministic automaton as input.
Since all of the conditions have been formalized under the condition $\vcsubset{\RR}$, 
\rObs{automata} can only be applied if both TRSs satisfy this variable condition.
Moreover, grounding is an essential preprocessing step, since tree automata only accept ground
terms.

\begin{example}
\label{aut ex}
Let $\RR_5 = \{\m a \to \m b_1, \m a \to \m b_2, x \to \m f(x)\}$.
Non-confluence can easily be shown since the critical pair $(\m b_1, \m b_2)$ is not joinable:
Take the automata $\AAA_i = (\{1\}, \FF, \{\m f(1) \to 1, \m b_i \to 1\}, \{1\})$, which
satisfy all conditions of \rObs{automata}. 
\end{example}

\section{Modularity of Confluence}

In \cite{Toyama87} it was proven that confluence is a modular property
for disjoint unions of TRSs. Whereas a certificate for applying this proof technique
is trivial by just providing the decomposition, we cannot certify these proofs, since
currently a formalization of this modularity result is missing.

However, we at least formalized the easy direction of the modularity theorem that
non-confluence of one of the TRSs implies non-confluence of the disjoint union,
and we can thus certify non-confluence proofs in a modular way. We base our 
certifier on the following theorem.
Here, we assume an infinite set of symbols\footnote{Note that in \isafor function symbols do not come with a fixed arity.} 
and finite signatures $\FF(\RR)$ and $\FF(\SSS)$ of the TRSs. 

\begin{theorem}
\label{mod}
Let $\FF(\RR) \cap \FF(\SSS) = \varnothing$, let $\vcsubset\RR$, let $\vclhs\SSS$.
Then $\nCR(\RR)$ implies $\nCR(\RR \cup \SSS)$.
\end{theorem}

\begin{proof}
By assuming $\nCR(\RR)$ there are $s,t,u$ such that $s \to_\RR^* t$, $s \to_\RR^* u$,
and $\njo\RR tu$. Since $\FF(\RR) \cap \FF(\SSS) = \varnothing$, w.l.o.g.\ we assume 
$\FF(s) \cap \FF(\SSS) = \varnothing$.\footnote{Here is exactly the
point where in the formalization we use the assumptions of finite signatures and an infinite
set of symbols. Then it is always possible to rename all symbols in $\FF(s) \cap \FF(\SSS)$
into fresh ones.} By $\vcsubset\RR$ we conclude that also $(\FF(t) \cup \FF(u))
\cap \FF(\SSS) = \varnothing$ must hold. Assume that $t$ and $u$ are joinable by
$\RR \cup \SSS$. By looking at the function symbols and using $\vclhs\SSS$ we conclude
that the joining sequences cannot use any rule from $\SSS$. Hence, $t$ and $u$ are 
joinable by $\RR$, a contradiction to $\njo\RR tu$. 
\end{proof}

There is an asymmetry in the modularity theorem, namely that $\RR$ and $\SSS$
have to satisfy different variable conditions. Note that in general it is not possible
to weaken these conditions as can be seen by the following two examples of \cite[Example 20 and 
example in Section 5.3]{IJCAR}. 
If $\RR = \{\m a \to \m b, \m a \to \m c\}$ and $\SSS = \{x \to \m d\}$ 
(or if $\RR = \{\m f(x,y) \to \m f(z,z), \m f(\m b, \m c) \to \m a, \m b \to \m d, \m c \to \m d\}$
and $\SSS = \{\m g(y,x,x) \to y, \m g(x,x,y) \to y\}$) then $\nCR(\RR)$, but $\CR(\RR \cup \SSS)$.
Hence $\vclhs\SSS$ (or $\vcsubset\RR$) cannot be dropped from \rTH{mod}.
The relaxation on the variable conditions sometimes is helpful:

\begin{example}
Consider the non-confluent $\RR_5$ of \rEX{aut ex} and $\SSS = \{\m g(x) \to y\}$.
By \rTH{mod} and $\nCR(\RR_5)$ we immediately conclude $\nCR(\RR_5 \cup \SSS)$.
Note that the proof in \rEX{aut ex} is not applicable to $\RR_5 \cup \SSS$, since
$\vcsubset{\RR_5 \cup \SSS}$ does not hold.
\end{example}

\paragraph{Acknowledgments}
We thank Thomas Sternagel for his formalized breadth-first search algorithm,
and Bertram Felgenhauer and Harald Zankl for integrating 
CPF-export into \CSI. The authors are listed in alphabetical order regardless of individual
contributions or seniority.

\bibliographystyle{plain}
\bibliography{references}
\end{document}